%% file: main.tex
\newif\iflipics
\lipicsfalse 

\iflipics
\documentclass[a4paper,UKenglish,cleveref, autoref, thm-restate, authorcolumns]{lipics-v2019}
\else
\documentclass[11pt]{article}
\fi

\input{macros.tex}

\iflipics
\title{New Verification Schemes for Frequency-Based Functions on Data Streams} 

\author{Prantar Ghosh}{Dartmouth College, USA}{prantar.ghosh.gr@dartmouth.edu}{}{Work supported in part by NSF under award CCF-1907738}

\authorrunning{P.~Ghosh} 

\Copyright{Prantar Ghosh} 



\category{} 

\relatedversion{} 

\supplement{}

\acknowledgements{The author would like to thank Amit Chakrabarti and Justin Thaler for several helpful discussions. He is also grateful to the anonymous FSTTCS 2020 reviewers for their valuable comments, especially to the reviewer who gave a sketch of a scheme using a randomized estimation algorithm, pointing out that it can handle longer streams and that determinism isn't strictly necessary for the subroutine.}

\EventEditors{John Q. Open and Joan R. Access}
\EventNoEds{2}
\EventLongTitle{42nd Conference on Very Important Topics (CVIT 2016)}
\EventShortTitle{CVIT 2016}
\EventAcronym{CVIT}
\EventYear{2016}
\EventDate{December 24--27, 2016}
\EventLocation{Little Whinging, United Kingdom}
\EventLogo{}
\SeriesVolume{42}
\ArticleNo{23}

\else
\title{New Verification Schemes for Frequency-Based Functions on Data Streams} 

\author{Prantar Ghosh\thanks{Department of Computer Science, Dartmouth College. Email: prantarg@cs.dartmouth.edu. Work supported in part by NSF under award CCF-1907738.}}
\date{}
\fi
\begin{document}

\maketitle

\begin{abstract}
We study the general problem of computing {\em frequency-based functions}, i.e., the sum of any given function of data stream frequencies. Special cases include fundamental data stream problems such as computing the number of distinct elements ($F_0$), frequency moments ($F_k$), and heavy-hitters. It can also be applied to calculate the maximum frequency of an element ($F_{\infty}$). 

Given that exact computation of most of these special cases provably do not admit any sublinear space algorithm, a natural approach is to consider them in an enhanced data streaming model, where we have a computationally unbounded but untrusted {\em prover} sending {\em proofs} or help messages to ease the computation. Think of a memory-restricted client delegating the computation to a stronger cloud service whom it doesn't want to trust blindly. Using its limited memory, it wants to {\em verify} the proof that the cloud sends. Chakrabarti et al.~(ICALP '09) introduced this setting as the {\em annotated data streaming model} and showed that multiple problems including exact computation of frequency-based functions---that have no sublinear algorithms in basic streaming---do have annotated streaming algorithms, also called {\em schemes}, with both space and proof-length sublinear in the input size.

We give a general scheme for computing any frequency-based function with both space usage and proof-size of $O(n^{2/3}\log n)$ bits, where $n$ is the size of the universe. This improves upon the best known bound of $O(n^{2/3}\log^{4/3} n)$ given by the seminal paper of Chakrabarti et al.~and as a result, also improves upon the best known bounds for the important special cases of computing $F_0$ and $F_{\infty}$. We emphasize that while being quantitatively better, our scheme is also qualitatively better in the sense that it is simpler than the previously best scheme that uses intricate data structures and elaborate subroutines. Our scheme uses a simple technique tailored for this model: the verifier solves the problem partially by running an algorithm known to be helpful for it in the basic (sans prover) streaming model and then takes the prover's help to solve the remaining part. 
\end{abstract}

\section{Introduction}\label{sec:intro}
Interactive proof systems have contributed a very important conceptual message to computer science: it is possible for a computationally bounded entity to reduce its computational cost for a problem if it is only required to verify a proof of the solution instead of finding a solution on its own. This concept led to celebrated results such as IP = PSPACE \cite{Shamir92} and the PCP Theorems \cite{AroraLMSS98, AroraS98}. It is natural to incorporate this idea to deal with challenging problems in massive data streams so as to reduce the impractical computational costs for such problems. This incorporation led to the following setting: a space-restricted client reading a huge data stream outsources the computation to a more powerful entity, such as a cloud service, with unbounded space. The cloud sends the result of the computation to the client who refuses to blindly trust it since it might be malicious or might have incurred some hardware failure. Therefore, the cloud (henceforth named ``Prover'') also sends the client (henceforth named ``Verifier'') a {\em proof} in support of its results. Verifier needs to use his limited space to collect sufficient information from the stream so as to verify the proof. In the case that Prover is honest, Verifier can use it as a help message to find the solution to the underlying problem. Otherwise, he rejects the proof. This combination of data streaming with prover-verifier systems has been fruitful: multiple works \cite{AbdullahDRV16, ChakrabartiCGT14, ChakrabartiCMT14, ChakrabartiCMTV15, ChakrabartiG19, ChakrabartiGT20, CormodeMT13, CormodeTY11, KlauckP13, KlauckP14, Thaler16} have shown that several intractable problems in the basic data streaming model turn out to be solvable in prover-enhanced models using verification space and proof-length sublinear in the input size. 

Chakrabarti et al.~\cite{ChakrabartiCMT14} formally defined this enhanced data streaming model as the {\em annotated data streaming model}. An algorithm in this model is called a {\em scheme}. In designing a scheme, the two important complexity parameters that we need to focus on are the {\em space} used by Verifier and the {\em size of the proof} sent by Prover. A scheme that has a proof-length of $O(h)$ bits and uses $O(v)$ bits of space is called an {\em $(h,v)$-scheme}.

Since its inception, data streaming algorithms have been extensively studied for fundamental statistical problems such as counting the number of distinct elements in a stream $(F_0)$ \cite{AlonMS99, BarYossefJKST02, FlajoletM85, KaneNW10-pods}, the $k$th frequency moment for $k>0$ ($F_k$) \cite{AlonMS99, GangulyC07, IndykW05, Woodruff04}, the maximum frequency of an element ($F_{\infty}$) \cite{AlonMS99, KaneNW10}, and the $\ell_p$-norm of the frequency vector for some $p\geq 0$ \cite{KaneNW10, KaneNW10-pods, NelsonNW12}. All of these problems are special cases of (or can be solved by easily applying) the general problem of computing {\em frequency-based functions}: given a function $g:\ZZ\to \ZZ^+$, find $\sum_{j=1}^n g(f_j)$, where, for each $j$ in the universe $\{1,\ldots,n\}$, $f_j$ is the frequency of the $j$th element. This general problem was notably addressed by the celebrated seminal paper by Alon, Matias, and Szegedy \cite{AlonMS99}: they asked for a characterization of precisely which frequency-based functions can be approximated efficiently in the basic streaming model. The aforementioned paper by Chakrabarti et al. \cite{ChakrabartiCMT14} studied such statistical problems in the annotated streaming setting and gave several interesting schemes. In particular, for the general problem of computing frequency-based functions, they gave an $(n^{2/3}\log^{4/3} n, n^{2/3}\log^{4/3} n)$-scheme. Their scheme uses an intricate data structure with binary trees and calls upon a subroutine for heavy-hitters that uses an elaborate framework called {\em hierarchical heavy hitters}. 

Given how general the problem is, with several important special cases having numerous applications, it is important and beneficial to have a {\em simple} scheme for the general problem. In this work, we design such a simple scheme that uses the most basic and classical data structure for frequency estimation: the Misra-Gries summary \cite{MisraG82}. Our scheme ends up improving the best known complexity bounds for the problem: we give an $(n^{2/3}\log n, n^{2/3}\log n)$-scheme. No better bounds or simpler algorithms were known even for the special cases of computing $F_0$ or $F_{\infty}$. Our result thus simplifies and improves the bounds for these problems as well. 

The aforementioned scheme works for streams of length $m=O(n)$, an assumption that was also made by Chakrabarti et al. \cite{ChakrabartiCMT14}. However, their scheme can be made to work for longer turnstile streams as long as $\Vert \bff \Vert_1=O(n)$. We show how to use the Count-Median Sketch \cite{CormodeM05}, an estimation algorithm with stronger guarantees than Misra-Gries, to get a scheme with similar complexity bounds for these long streams. But since the Count-Median Sketch is randomized (contrary to Misra-Gries), we incur a non-zero completeness error for this scheme. The high-level idea of both our schemes is the following: we use the estimation algorithm as a primitive to ``partially'' solve the problem. Prover then helps Verifier with the ``remaining'' unsolved part.
\subsection{Setup and Terminology}\label{sec:setup}
We formalize the setting described above.  A {\em scheme} for computing a function $g(\sigma)$ of the input stream $\sigma$ is a triple $(\cH, \cA, \out)$, where $\cH$ is a function that Prover uses to generate the help message or proof-stream for $\sigma$, given by $\cH(\sigma)$, $\cA$ is a data streaming algorithm that Verifier runs on the stream $\sigma$ using a random string $R$ to produce a summary $\cA_R(\sigma)$, and $\out$ is a streaming algorithm that Verifier runs on the proof-stream $\cH(\sigma)$ and also uses $\cA_R(\sigma)$ and $R$ to generate an output $\out_R(\cH(\sigma), \cA_R(\sigma))$ in range$(g)\cup \bot$, where the symbol $\bot$ denotes rejection of the proof. Note that if the proof-length $|\cH(\sigma)|$ is larger than the memory of Verifier, then he needs to process $\cH(\sigma)$ as a stream. 

A scheme $(\cH, \cA, \out)$ has completeness error $\eps_c$ and soundness error $\eps_s$ if it satisfies 
\begin{itemize}[topsep=4pt,itemsep=0pt]
  \item (completeness) $\forall \sigma: \Pr_R[\out_R(\cA_R(\sigma),\cH(\sigma)) = g(\sigma)] \ge 1-\eps_c$; 
  \item (soundness) $\forall \sigma, H: \Pr_R[\out_R(\cA_R(\sigma),H) \notin \{g(\sigma), \bot\}] \le \eps_s$.
\end{itemize}

Informally, this means that an honest Prover can convince Verifier to produce the correct output with high probability. Again, if Prover is dishonest, then, with high probability, Verifier rejects the proof. We usually aim for $\eps_c, \eps_s \leq 1/3$ (they can be boosted down using standard techniques incurring a small increase in the space usage). A scheme is said to have {\em perfect completeness} if $\eps_c=0$.

The {\em hcost} (short for ``help cost'') of a scheme $(\cH, \cA, \out)$ is defined as $\max_{\sigma} |\cH(\sigma)|$, i.e., the maximum number of bits required to express a proof. The {\em vcost} (short for ``verification cost'') is the maximum bits of space used by the algorithms $\cA_R(\sigma)$ and $\out_R(\sigma)$, where the maximum is taken over all inputs $\sigma$ and possible random strings $R$. A scheme with hcost $O(h)$ and vcost $O(v)$ is called an $(h,v)$-scheme. An $(h,v)$-scheme is interesting if $h>0$ and $v$ is asymptotically smaller than the best bound achievable for $h=0$, i.e., in the basic (sans prover) streaming model.  

\subsection{Our Results and Techniques}\label{sec:restech}
In this section, we state our results and give an overview of our techniques. 

\mypar{Results} Given a stream with elements in $[n]$, let $\bff$ denote its frequency vector $\langle f_1,f_2,\ldots,f_n \rangle$, where $f_j$ is the frequency of the $j$th element. A \emph{frequency-based function} is a function $G(\bff)$ of the form $G(\bff):=\sum_{j=1}^n g(f_j)$ for some function $g: \ZZ \to \ZZ^+$.

Our main result is captured in the following theorem which we prove in \Cref{sec:EMG-Scheme}.
\begin{restatable}{theorem}{thmmain}\label{thm:mainthm}
There is an $(n^{2/3}\log n,n^{2/3}\log n)$-scheme for computing any frequency-based function in any turnstile stream of length $m=O(n)$. The scheme is perfectly complete and has soundness error at most $1/\text{poly(n)}$.
\end{restatable}

With some modifications, we obtain a similar scheme for longer streams at the cost of imperfect completeness. This is given by the following theorem which we prove in \Cref{sec:CntSkcScheme}.
\begin{restatable}{theorem}{thmsktchprot}\label{thm:sktchprot}
There is an $(n^{2/3}\log n,n^{2/3}\log n)$-scheme for computing any frequency-based function in any turnstile stream with $\Vert\bff \Vert_1=O(n)$. The scheme has completeness and soundness errors at most $1/3$.
\end{restatable}
As a consequence, we get schemes with the same complexity bounds for the problems of computing $F_0$, $F_{\infty}$, and checking multiset inclusion (see \Cref{cor:multi} for formal definition). Just as for frequency-based functions, our schemes also improve upon the best known bounds for these special cases and applications\footnote{Computing $F_k$ for constant $k>0$ is a well-studied special case for which better bounds are known \cite{ChakrabartiCMT14}.}. We discuss these results in detail in \Cref{sec:apps}.  
\begin{restatable}{corollary}{fzero}\label{cor:f0}
For any turnstile stream with $\Vert \bff \Vert_1=O(n)$, there is an $(n^{2/3}\log n,n^{2/3}\log n)$-scheme for computing $F_{0}$, the number of distinct elements with non-zero frequency, with completeness and soundness errors at most $1/3$. The scheme can be made perfectly complete with soundness error $1/\text{poly}(n)$ if the stream has length $m=O(n)$. 
\end{restatable}
\begin{restatable}{corollary}{finf}\label{cor:finf}
For any turnstile stream with $\Vert \bff \Vert_1=O(n)$, there is an $(n^{2/3}\log n,n^{2/3}\log n)$-scheme for computing $F_{\infty}$, the maximum frequency of an element, with completeness and soundness errors at most $1/3$. The scheme can be made perfectly complete with soundness error $1/\text{poly}(n)$ if the stream has length $m=O(n)$.
\end{restatable}
\begin{restatable}{corollary}{multiset}\label{cor:multi}
Let $X,Y \subseteq [n]$ be multisets of size $O(n)$. Given a stream where elements of $X$ and $Y$ arrive in interleaved manner, there is an $(n^{2/3} \log n, n^{2/3} \log n)$-scheme for determining whether $X\subseteq Y$.
\end{restatable}
\mypar{Techniques} 
Computing frequency-based functions is challenging simply because we don't have enough space to store all the exact frequencies. However, there are efficient small-space algorithms---e.g., Misra-Gries algorithm \cite{MisraG82}, Count-Median Sketch \cite{CormodeM05}---that return reasonably good {\em estimates} of the frequencies. We use such an algorithm as a primitive in our schemes. The estimates returned partially solve the problem by helping us identify the ``heavy-hitters'' or the most frequent items. There cannot be too many heavy-hitters and hence, the all-powerful Prover can send Verifier the exact frequencies of these elements (which of course need to be verified) without too much communication. On the other hand, the rest of the elements, though large in number, have relatively small frequency. We show a way to encode the answer in terms of a low-degree polynomial when the frequencies are small. Prover can then send us this polynomial using few bits, enabling us to solve the problem with small communication overall. 

We remark that the high-level technique used in our first scheme---using Misra-Gries as a subroutine---might be more widely applicable than that used in the second one, i.e., using Count-Median Sketch. This is because Misra-Gries is deterministic while Count-Median is randomized. In general, both Prover and Verifier can locally run a {\em deterministic} algorithm on the input, and then, Prover can send messages based on the final state of that algorithm. Note that it isn't clear if a {\em randomized} algorithm can always help in this regard since we assume that Prover and Verifer do not have access to shared randomness\footnote{This assumption is made so that it corresponds to the MA communication model. Access to shared randomness corresponds to the AMA communication model where better bounds are known \cite{GurR13}.}. Hence, the final states of the algorithm might vary drastically for Prover and Verifier if they run it locally with their own private randomness. For our problem, we don't run into this issue since we don't require Prover to know the exact output of the Verifier's local estimation algorithm.

Other techniques used are pretty standard in this area. We use techniques based on the famous {\em sum-check protocol} of Lund et al.~\cite{LundFKN92} that encodes answers as sum of low-degree polynomials. In our case, where Prover sends only a single message to Verifier, a quantity of interest is expressed as the sum of evaluations of a low-degree univariate polynomial. Since the polynomial has low-degree, it can be expressed with a small number of monomials. Thus, Prover needs only a few bits to express the set of coefficients that describe the polynomial, leading to short proof-length. Moreover, to verify the authenticity of the polynomial, Verifier needs to evaluate it at just a single random point, the space for which he can afford. The main challenge in this technique is to find the proper low-degree polynomials to encode the answer, and in this work, we give such new polynomial encodings for the underlying sub-problems. Another standard technique we use is the {\em shaping technique} that transforms a one-dimensional vector into a two-dimensional array. On a high-level, this helps in ``distributing'' the work between Prover and Verifier as they each ``take care of'' a single dimension. Pertaining to the streaming model, we exploit the popular technique of {\em linear sketching} where we express a quantity of interest as a linear combination of the stream updates, which helps us to maintain the quantity dynamically as the stream arrives.
\subsection{Related Work}\label{sec:related}
Early works on the concept of stream outsourcing and verification were done by the database community \cite{LiYHK07, PapadopoulosYP07, Tucker05, YiLHKS08}. Motivated by these works, Chakrabarti et al.~\cite{ChakrabartiCMT14} abstracted out and formalised the theoretical aspects of the settings. They defined two types of stream verification settings: (i) the {\em annotated data streaming} setting---calling the schemes as {\em online schemes}---where Prover and Verfier read the input stream together and Prover sends help messages during and/or after the stream arrival based on the part of the stream she has seen so far, and (ii) the {\em prescient} setting where Prover knows the entire stream upfront, i.e., before Verifies sees it, and can send help messages accordingly. Several subsequent works \cite{ChakrabartiCGT14, ChakrabartiG19, ChakrabartiGT20, CormodeMT13, KlauckP13, Thaler16} studied these non-interactive models. Natural generalizations of the model, where we allow multiple rounds of interaction between Prover and Verifier, have also been explored. These include {\em Arthur-Merlin streaming protocols} (Prover is named ``Merlin'' and Verifier is named ``Arthur'' following a long-standing tradition in complexity theory) of Gur and Raz \cite{GurR13} and the {\em streaming interactive proofs} (SIP) of Cormode et al.~\cite{CormodeTY11}. The latter setting was further studied by multiple works \cite{AbdullahDRV16, ChakrabartiCMTV15, KlauckP14}. We refer the reader to the expository article by Thaler
\cite{Thaler-encyclopedia} for a detailed survey of this area.

We state the results with the standard assumption \cite{ChakrabartiCMT14, CormodeTY11} that the stream length $m=O(n)$. Chakrabarti et al.~\cite{ChakrabartiCMT14} gave two schemes for computing any general frequency-based function: an online $(n^{2/3}\log^{4/3} n,n^{2/3}\log^{4/3} n)$-scheme and a prescient $(n^{2/3}\log n, n^{2/3}\log n)$-scheme. They noted that the schemes apply to get best known schemes for the special cases of computing the number of distinct elements ($F_0$), the maximum frequency ($F_{\infty}$), and for checking multiset inclusion. They also showed a lower bound that any online or prescient $(h,v)$-scheme for the problem (even for the aforementioned special cases) requires $hv\geq n$. They designed schemes with better bounds for certain other frequency-based functions, often matching this lower bound up to polylogarithmic factors. For instance, for any $hv=n$, they gave an online $(k^2h\log n, kv\log n)$-scheme for calculating the $k$th frequency moment $F_k$ for any positive integer $k$, and a $(\phi^{-1}\log^2 n+h\log n, v\log n)$-scheme for computing the $\phi$-heavy hitters (elements with frequency of at least a $\phi$-fraction of the stream length). 

The specific problem of computing $F_0$ has been studied by multiple works in various stream verification models. Cormode, Mitzenmacher, and Thaler \cite{CormodeMT10} studied the problem in the stronger SIP-model and gave a $(\log^3 n,\log^2 n)$-SIP with $O(\log^2 n)$ rounds of communication. For the case where we restrict the number of rounds to $O(\log n)$, Cormode, Thaler, and Yi \cite{CormodeTY11} gave a $(\sqrt{n}\log^2 n, \log^2 n)$-SIP. Klauck and Prakash \cite{KlauckP14} improved this to a $(\log^4 n \log \log n, \log^2 n \log \log n)$-SIP. Gur and Raz \cite{GurR13}
designed an $(\tO(\sqrt{n}), \tO(\sqrt{n}))$-AMA-streaming protocol\footnote{AMA stands for the communication pattern Arthur-Merlin-Arthur} (the $\tO(\cdot)$ notation hides polylog$(n)$ factors) for $F_0$.


\section{Preliminaries}\label{sec:prelim}
Here, we discuss the streaming models we study and some standard results that we use in our schemes.

Throughout this paper, the stream elements come from the universe $[n]:=\{1,\ldots,n\}$ and the stream length is $m$. In the {\em turnstile} streaming model, tokens are of the form $(j,\Delta)\in [n]\times \ZZ$, which means $\Delta$ copies of the element $j$ are inserted (resp. deleted) if $\Delta>0$ (resp. $\Delta<0$). The \emph{cash register} or {\em insert-only} streaming model is the special case when $\Delta$ is always positive. In this paper, for simplicity, we assume unit updates, i.e., $\Delta\in \{-1,1\}$ always. The assumption can be easily removed by looking at an update as a collection of multiple unit updates. 

For a stream $\sigma = \langle (a_1, \Delta_1),\ldots, (a_m, \Delta_m)\rangle$, the \emph{frequency vector} $\bff(\sigma)$ is defined as $\langle f_1,\ldots,f_n\rangle$ where $f_j$ is the {\em frequency} of element $j$, given by $f_j:=\sum_{\substack{i\in[m]:\\a_i = j}} \Delta_i$. We denote {\em estimates} of $f_j$ by $\hf_j$. We drop the argument $\sigma$ when the stream is clear from the context. 

In our schemes, we use the standard technique of \emph{sketching} a frequency vector by evaluating its \emph{low-degree extension} at a random point. We explain what this means. We transform (or {\em shape}) our frequency vector of length $n$ into a $2$-dimensional $d_1\times d_2$ array~$f$, where $d_1d_2=n$, using some canonical bijection from $[n]$ to $[d_1]\times [d_2]$. This means that the domain of the function $f$ can now be seen as $[d_1]\times [d_2]$. We work on a finite field $\FF$ with large enough characteristic such that the values don't ``wrap around'' under operations in $\FF$. By Lagrange's interpolation, there is a unique polynomial $\tf(X, Y) \in \FF[X, Y]$ with $\deg_X(\tf)=d_1-1$ and $\deg_Y(\tf)=d_2-1$ such that $\tf(x,y)=f(x,y)$ for all $(x,y)\in [d_1]\times [d_2]$. We call $\tf$ the {\em low-degree  $\FF$-extension} of~$f$. For each $(x,y)\in [d_1]\times[d_2]$, we have ``Lagrange basis polynomials'' defined as
\begin{align} \label{eq:unit-impulse}
\delta_{x, y}(X, Y)
    := \left(\prod_{x_i \in [d_1] \setminus \{x\}} \frac{X - x_i}{x - x_i}\right)\cdot\left( \prod_{y_i \in [d_2] \setminus \{y\}} \frac{Y - y_i}{y - x_i}\right)
\end{align}
We can write $\tf$ as a linear combination of these polynomials as follows:
\[\tf(X,Y)
 = \sum_{(x,y) \in [d_1] \times [d_2]}
 f(x, y)\, \delta_{x,y}(X, Y)\]

In particular, if $f$ is built up from a stream of turnstile updates $\langle ((x,y)_j, \Delta_j)\rangle$, then

\begin{equation} \label{eq:stream-update}
  \tf(X, Y)
    = \sum_j \Delta_j\,
    \delta_{(x,y)_j}(X, Y) \,.
\end{equation}
Thus, we can use \cref{eq:stream-update} to {\em maintain} $\tf(x^*,y^*)$ at some fixed point $(x^*,y^*)$ dynamically with stream updates. We formalize this in the following fact.

\begin{fact} \label{fact:dynamicupdate}
  Given a point $(x^*,y^*) \in \FF^2$ and a stream of updates to an initially-zero $d_1\times d_2$-dimensional array $f$, we can maintain $\tf(x^*,y^*)$ using $O(\log|\FF|)$ space. For implementation details and generalizations, see Cormode et al.~\cite{CormodeTY11}.
\end{fact}

An important subroutine in one of our schemes
s the classic Misra-Gries algorithm for frequency estimation \cite{MisraG82} which, given an input stream of $m$ elements and a fraction $\phi$, estimates the frequency of the stream elements within an additive factor of $\phi m$. We recall this algorithm in \Cref{algo:mgclass}.

Informally, the algorithm does the following: it keeps an array or ``dictionary'' $K$ indexed by ``keys'' that are elements of the stream and each of them has an associated counter $K[i]$. At any point of time, the array has at most $\ceil{\phi^{-1}}$ keys. When a stream element arrives, it increments the counter for the element if it's present in the keys (it includes it in the keys if there are less than $\ceil{\phi^{-1}}$ keys), and otherwise decrements the counter of every key. If a counter for a key becomes $0$, it is removed from $K$. Finally, the estimate $\hf_j$ is given by $K[j]$ (which is $0$ if $j$ is not in the keys). The guarantees of the algorithm is given in \Cref{fact:guaranteemg}.

\begin{algorithm}\caption{\cite{MisraG82} Misra-Gries algorithm for frequency estimates in insert-only streams}\label{algo:mgclass}
\begin{algorithmic}[1]
\Require Stream $\sigma$; $\phi\leq 1$
\State Initialize $K\leftarrow$ empty array

\hspace{-1.3cm}\textbf{Process}(token $j\in \sigma$):
\If{$j\in \text{keys}(K)$}
\State $K[j]\leftarrow K[j]+1$
\Else 
\If{$|\text{keys}(K)|< \ceil{\phi^{-1}}$}
\State $K[j]\leftarrow 1$
\Else  
\For{$i \in \text{keys}(K)$}:
\State $K[i]\leftarrow K[i]-1$
\State if $K[i]=0$ then remove $i$ from keys$(K)$
\EndFor
\EndIf
\EndIf

\hspace{-1.3cm}\textbf{Output}:
\For{$j\in [n]$}:
\State if $j\in \text{keys}(K)$ then \Return $\hf_j=K[j]$; else \Return $\hf_j=0$
\EndFor
\end{algorithmic}
\end{algorithm}

\begin{fact}[\cite{MisraG82}]\label{fact:guaranteemg}
For an insert-only stream of $m$ elements in $[n]$, given any $\phi\leq 1$, \Cref{algo:mgclass} uses $O(\phi^{-1}(\log n+\log m))$ space and returns frequency estimates $\langle\hf_j: j\in [n]\rangle$ such that, for all tokens $j\in [n]$, we have $f_j-\phi m\leq \hf_j\leq f_j$.
\end{fact}
Note that this algorithm was designed for insert-only streams and doesn't work for turnstile streams. To use it for turnstile streams, we need to make appropriate modifications (which we do in \Cref{sec:MisraG}).  

\section{Computing Frequency-based Functions in Turnstile Streams}\label{sec:mainsec}

Let $\bff$ be the frequency vector of a stream as defined in \Cref{sec:prelim}. Recall that a \emph{frequency-based function} is a function $G(\bff)$ of the form $G(\bff):=\sum_{j\in [n]} g(f_j)$ for some function $g: \ZZ \to \ZZ^+$. In this section, we obtain an improved $(n^{2/3}\log n, n^{2/3}\log n)$-scheme for computing any frequency-based function for some predetermined function $g$.
As stated earlier, we design a scheme exploiting the Misra-Gries algorithm (\Cref{algo:mgclass}). We want to use it as a subroutine in our problem for turnstile streams, but it works only in the insert-only model. Therefore, in \Cref{sec:MisraG}, we provide a simple extension of the algorithm that attains a similar guarantee for turnstile streams. In \Cref{sec:mainalg}, first, we use this extended Misra-Gries (EMG) algorithm as a subroutine for our scheme for computing frequency-based functions. Next, we show that we can instead use the Count-Median Sketch \cite{CormodeM05} to make it work for longer streams. In \Cref{sec:apps}, we discuss some important applications of our schemes.  

\subsection{Extension of Misra-Gries Algorithm for Turnstile Streams}\label{sec:MisraG}

The extended Misra-Gries algorithm (henceforth called ``EMG algorithm'') works as follows: we process the positive and negative updates separately in two parallel copies of \Cref{algo:mgclass} to estimate the total positive update and (absolute value of) the total negative update. In the second copy, we can actually think of the updates as ``increments'' since only negative updates are processed there. Thus, what we are actually estimating is the absolute value of the total negative update. 

For each $j$, let the total positive update be $f_j^+$ and (absolute value of) the total negative update $f_j^-$. Then, the actual frequency is $f_j=f_j^+-f_j^-$. Denote the corresponding estimates given by the copies of \Cref{algo:mgclass} by $\hf_j^+$ and $\hf_j^-$. Then $\hf_j:=\hf_j^+-\hf_j^-$ gives a similar guarantee as \Cref{fact:guaranteemg} for turnstile streams; this time, we also incur an additive error of $\phi m$ on the upper bound. 






To see this, note that by \Cref{fact:guaranteemg}, we have, $\forall j\in [n]$, 
\begin{align}
  f_j^+-\phi m\leq \hf_j^+\leq f_j^+ \label{eq:pos} \\ 
    f_j^--\phi m\leq \hf_j^-\leq f_j^-\label{eq:neg} 
\end{align}

 Thus, \cref{eq:pos,eq:neg} give
 $f_j^+ - f_j^- - \phi m \leq \hf_j^+-\hf_j^- \leq f_j^+ - f_j^- + \phi m$,
 i.e.,
\begin{equation}\label{eq:extmg}
    f_j - \phi m \leq \hf_j \leq f_j + \phi m
\end{equation}

Hence, this time we get double sided error. This estimate would suffice for getting our desired scheme. Therefore, we get the following lemma.

\begin{lemma}\label{lem:1}
Given a turnstile stream of $m$ elements in $[n]$, the EMG algorithm uses $O(\phi^{-1}(\log n + \log m))$ space and returns a summary $\langle \hf_j: j\in [n]\rangle$ such that, for all $j\in [n]$, we have $f_j - \phi m \leq \hf_j \leq f_j + \phi m$.
\end{lemma}

\begin{remark}
The guarantee given by the EMG algorithm may not be very useful in general for turnstile streams. This is because the total number of stream updates $m$ can be huge, whereas the frequency of each token can be small since we allow both increments and decrements in the turnstile model. The classic Misra-Gries algorithm for insert-only model, on the other hand, has a good guarantee (\Cref{fact:guaranteemg}) since $m=\Vert \bff \Vert_1$ in this model. However, for our purpose, the guarantee in \Cref{lem:1} is good enough since we assume that $m=O(n)$.      
\end{remark}

\subsection{Schemes for Frequency-based Functions}\label{sec:mainalg}

First, in \Cref{sec:EMG-Scheme}, we describe a protocol for computing frequency-based functions in turnstile streams of length $O(n)$ and prove \Cref{thm:mainthm}. Next, in \Cref{sec:CntSkcScheme}, we show that the scheme can be modified to work for any turnstile stream with $\Vert \bff \Vert_1=O(n)$, proving \Cref{thm:sktchprot}. The completeness error in the latter scheme is, however, non-zero. 
\subsubsection{Perfectly Complete Scheme for \boldmath{$O(n)$}-Length Streams}\label{sec:EMG-Scheme}
As in prior works \cite{ChakrabartiCMT14, CormodeTY11}, we solve the problem for stream length $m=O(n)$. Hence, by \Cref{lem:1}, the EMG algorithm takes $O(\phi^{-1}\log n)$ space and gives, for some constant $c$,  
\begin{equation}\label{eq:cnbasic}
    \forall j\in [n]:~f_j-\phi cn\leq \hf_j\leq f_j+\phi cn\,.
\end{equation}
Set $\phi=(cn^{2/3})^{-1}$. Therefore, we have an $O(n^{2/3}\log n)$ space algorithm that guarantees
\[\forall j\in [n]:~f_j-n^{1/3}\leq \hf_j\leq f_j+n^{1/3}\,.\]  
Let $K$ denote the set of keys in the final state of the EMG algorithm for the setting of $\phi=1/(cn^{2/3})$. Observe that if $\hf_j=0$ for some $j$ (i.e., $j\not\in K$), we know that $f_j\in [-n^{1/3},n^{1/3}]$.

Define $h(j) = \mathbb{I}\{j\not\in K\}$ where $\mathbb{I}$ is the $0$-$1$ indicator function. We have 
\[\sum_{j\in [n]} g(f(j)) = \sum_{j\in K} g(f(j)) + \sum_{j\not\in K} g(f(j)) = \sum_{j\in K} g(f(j)) + \sum_{j\in [n]} g(f(j))h(j)\]

Let $L:=\sum_{j\in K} g(f(j))$ and $R:= \sum_{j\in [n]} g(f(j))h(j)$. We shall compute $L$ and $R$ separately and add them to get the desired answer. 

We {\em shape} (see \Cref{sec:prelim}) the $1$D array $[n]$ into a $2$D $n^{1/3}\times n^{2/3}$ array. Thus, we get
\[R=\sum_{x\in [n^{1/3}]}\sum_{y\in [n^{2/3}]} g(f(x,y))h(x,y)\]

As is standard \cite{ChakrabartiCMT14}, we assume that the range of the function $g$ is upper bounded by some polynomial in $n$, say $n^p$. Pick a prime $q$ such that $n^{p+1}< q<2n^{p+1}$. We will work in the finite field $\FF_q$ and the upper bound on the range of $g$ ensures that $G(\bff)$ will not ``wrap around'' under arithmetic in $\FF_q$.

Let $\tf, \tilde{h}$ be polynomials of lowest degree over the finite field $\FF_q$ that agree with $f,h$ respectively at all values in $[n^{1/3}]\times [n^{2/3}]$. Note that, by Lagrange's interpolation, both $\tf$ and $\tilde{h}$ have degrees $n^{1/3}-1$ and $n^{2/3}-1$ in the two variables (see \Cref{sec:prelim}). Again, let $\tg$ denote the polynomial of lowest degree that agrees with $g$ at all values in  $[-n^{1/3},n^{1/3}]$. Thus, $\tg$ has degree $2n^{1/3}$.

Therefore, we have 
\[R=\sum_{x\in [n^{1/3}]}\sum_{y\in [n^{2/3}]} \tilde{g}(\tilde{f}(x,y))\tilde{h}(x,y)\]
i.e., we can write
\begin{equation}\label{eq:rsum}
    R=\sum_{x\in [n^{1/3}]}P(x)\,,
\end{equation}
where the polynomial $P$ is given by 
\begin{equation}\label{eq:ppoly}
    P(X) = \sum_{y\in [n^{2/3}]} \tilde{g}(\tilde{f}(X,y))\tilde{h}(X,y)
\end{equation} 

To compute $L$, it suffices to obtain the values $f_j$ for all $j\in K$ since $g$ is predetermined. In our protocol, Prover would send values $f'_j$ that she claims to be $f_j$ for all $j\in K$. Define
\[T:= \sum_{j\in K} (f_j - f'_j)^2\]

Note that we have $f_j=f'_j$ for each $j$ if and only if $T=0$. Set $f'_j:=0$ for all $j\not\in K$.
Thus, we can rewrite $T$ as 
\[T=\sum_{j\in [n]} (f_j - f'_j)^2(1-h(j))\]
Using shaping as before, we get
\begin{equation}
  T=\sum_{x\in [n^{1/3}]}\sum_{y\in [n^{2/3}]} (f(x,y) - f'(x,y))^2(1-h(x,y))  
\end{equation}
Let $\tf'$ denote the polynomial of lowest degree over $\FF_q$ that agrees with $f'$ at all values in $[n^{1/3}]\times [n^{2/3}]$. Therefore, we have \begin{equation}\label{eq:tsum}
    T=\sum_{x\in [n^{1/3}]} Q(x)
\end{equation}
where the polynomial $Q$ is given by
\begin{equation}\label{eq:qpoly}
  Q(X)=\sum_{y\in [n^{2/3}]}(\tilde{f}(X,y)-\tilde{f}'(X,y))^2(1-\tilde{h}(X,y))\,.  
\end{equation}
We are now ready to describe the protocol. 
\begin{protocol}
\item[Stream processing.] Verifier picks $r \in \FF_q$ uniformly at random. As the stream arrives, he maintains $\tf(r,y)$ for all $y\in [n^{2/3}]$ (\Cref{fact:dynamicupdate}). In parallel, he runs the EMG algorithm setting $\phi = (cn^{2/3})^{-1}$.  

\item[Help message.] Prover sends polynomials $P'$ and $Q'$, and values $f'_j$ for all $j\in K$. She claims that $P', Q', f'$ are identical to $P,Q,f$ respectively. The polynomials are sent as streams of their coefficients following some canonical order of their monomials. Verifier evaluates $P'(r)$ and $Q'(r)$ as the polynomials are streamed.

\item[Verification and output.] Looking at the final state of the EMG subroutine, Verifier constructs $\tilde{h}(r,y)$ for all $y\in [n^{2/3}]$ (he can treat the keys as a stream and use \Cref{fact:dynamicupdate}). Also, from the values $f'_j$, he constructs $\tf'(r,y)$ for all $y\in [n^{2/3}]$. The $O(n^{1/3})$-degree polynomial $\tg$ is computed and stored in advance (we need to evaluate $g$ at all points in $[-n^{1/3},n^{1/3}]$ and then use Lagrange interpolation to get $\tg$). 

Thus, Verifier can now use \cref{eq:ppoly} to compute $P(r)$ and \cref{eq:qpoly} to compute $Q(r)$. He checks whether $P(r)=P'(r)$ and $Q(r)=Q'(r)$. If the checks pass, he believes $P',Q'$ are correct. He further checks whether $\sum_{x\in [n^{1/3}]}Q'(x) = 0$, i.e., by \Cref{eq:tsum}, whether $T=0$. If so, he believes that $f'_j=f_j$ for all $j\in K$. Next, he computes $L=\sum_{j\in K}g(f'(j))$, and using \Cref{eq:rsum}, he computes $R=\sum_{x\in [n^{1/3}]} P'(x)$.  Finally, $L+R$ gives the answer. 

\medskip
\item[Error probability.] The correctness analysis follows along standard lines of sum-check protocols. The scheme is perfectly complete since it follows from above that we always output correctly if Prover is honest. For soundness, note that the protocol fails if either $P\neq P'$ or $Q\neq Q'$, but $P(r)=P'(r)$ and $Q(r)=Q'(r)$. Then, $r$ is a root of the non-zero polynomial $P-P'$ or $Q-Q'$. Since degree of $P-P'$ is $O(n^{2/3})$ and that of $Q-Q'$ is $O(n^{1/3})$, they have at most $O(n^{2/3})$ roots in total. Since $r$ is drawn uniformly at random from $\FF_q$, where $q>n^{p+1}$, the probability that $r$ is such a root is at most $O(n^{2/3})/n^{p+1} \leq  1/\text{poly}(n)$ for sufficiently large $n$. Thus, the soundness error is at most $1/\text{poly}(n)$.

\medskip
\item[Help and Verification costs.]
The polynomials $P$ and $Q$ have degree $O(n^{2/3})$ and $O(n^{1/3})$ respectively. Thus, it requires $O(n^{2/3}\log n)$ bits in total to express their coefficients since each coefficient comes from $\FF_q$ that has size $\text{poly}(n)$. Recall that for the setting of $\phi=(cn^{2/3})^{-1}$, there are $O(\phi^{-1})=O(n^{2/3})$ keys in the EMG algorithm. Prover sends $f'_j$ for each $j\in K$, and since each frequency is at most $m=O(n)$, this requires $O(n^{2/3}\log n)$ bits to communicate. Therefore, the total hcost is $O(n^{2/3}\log n)$.

As noted above, the invocation of EMG algorithm takes $O(n^{2/3}\log n)$ space. Verifier maintains $\tf(r,y)$ and stores the values $\tilde{h}(r,y)$ and $\tf'(r,y)$ for all $y\in [n^{2/3}]$. Each value is an element in $\FF_q$, and hence they take up $O(n^{2/3}\log n)$ space in total. The $O(n^{1/3})$-degree polynomial $\tg$ takes $O(n^{1/3}\log n)$ space to store. Hence, the total vcost is $O(n^{2/3}\log n)$. 
\end{protocol}
Thus, we have proved the following theorem.
\thmmain*

\subsubsection{Handling longer streams at the cost of imperfect completeness}\label{sec:CntSkcScheme} 
The scheme in \Cref{sec:EMG-Scheme} requires stream length $m=O(n)$. Note that a turnstile stream with massive cancellations can have length $m \gg n$, but $\Vert\bff\Vert_1$ can still be $O(n)$. Chakrabarti et al.~\cite{ChakrabartiCMT14} presented their scheme under the assumption of $m=O(n)$, but their scheme can be made to work for longer streams as long as $\Vert\bff\Vert_1 = O(n)$. We can modify our scheme to handle such streams as well without increasing the costs, but we no longer have perfect completeness. We give a sketch of this scheme below highlighting the modifications.

We cannot use the EMG algorithm anymore because it doesn't give a strong guarantee with respect to $\Vert\bff\Vert_1$ for turnstile streams. We use the Count-Median Sketch instead which gives the following guarantee.
\begin{fact}[Count-Median Sketch \cite{CormodeM05}]\label{fact:cntmed}
For all $\phi,\,\eps >0$, there exists an algorithm that, given a turnstile stream of elements in $[n]$ with $\Vert \bff \Vert_1 =O(n)$, uses $O(\phi^{-1} \log (\eps^{-1})\log n)$ space and returns frequency estimates $\langle\hf_j: j\in [n]\rangle$ such that, with probability at least $1-\eps$, for all tokens $j\in [n]$, we have $f_j-\phi \Vert \bff \Vert_1\leq \hf_j\leq f_j+\phi \Vert \bff \Vert_1$.
\end{fact}

If $\Vert \bff \Vert_1 \leq cn$ for some constant $c$, then setting $\phi=(4cn^{2/3})^{-1}$ and $\eps=1/4$, we get that there is an $O(n^{2/3}\log n)$ space algorithm that, with probability at least $3/4$, gives
\begin{equation}\label{eq:cntmedguarantee}
    \forall j \in [n]: f_j-n^{1/3}/4\leq \hf_j\leq f_j+ n^{1/3}/4
\end{equation}

For this protocol, redefine the set $K$ as $K:=\{j: |f_j| \geq n^{1/3}/2\}$. Prover sends a set $K'$ that she claims is identical to $K$. Let $M$ denote the set $\{j: |\hf_j|\geq 3n^{1/3}/4\}$. Verifier checks whether $M\subseteq K'$, and if the check passes, he computes $\sum_{j\in K'} g(f_j)$ and $\sum_{j\not\in K'} g(f_j)$ separately, similar to the earlier protocol, and adds them to obtain the answer. 

\mypar{Error probability} For completeness, note that if Prover is honest and $K'=K$, then with probability at least $3/4$, we have $M\subseteq K'$. To see this, observe that, by the guarantees of the Count-Median Sketch (\cref{eq:cntmedguarantee}), for all $j\in [n]$ with $|\hf_j| \geq 3n^{1/3}/4$, we have $|f_j|\geq n^{1/3}/2$ with probability at least $3/4$. The rest of the completeness analysis is as before, and hence, there is no additional completeness error. Thus, the total completeness error of the scheme is at most $1/4$.

For soundness, suppose that $K'\neq K$. By the guarantees of the Count-Median Sketch, for all $j\in [n]$ with $|f_j| \geq n^{1/3}$, we have $|\hf_j|\geq 3n^{1/3}/4$ with probability at least $3/4$. Thus, $\{j : |f_j| \geq n^{1/3}\}\subseteq M$. Hence, if the check $M\subseteq K'$ passes, then with probability at least $3/4$, we have $\{j : |f_j| \geq n^{1/3}\}\subseteq K'$. Thus, if $j\not\in K'$, we have $|f_j|<n^{1/3}$. Therefore, 
the computation of $\sum_{j\not\in K'} g(f_j)$ goes through as before. The additional soundness error is at most $1/\text{poly}(n)$ as analyzed earlier. Thus, the total soundness error of the protocol is at most $1/4+1/\text{poly}(n)<1/3$.    

\mypar{Help and Verification costs} Clearly, since $\Vert \bff \Vert_1 \leq cn$, we have $|K|= O(n^{2/3})$ which adds $O(n^{2/3}\log n)$ bits to the hcost. 
The Count-Median Sketch takes space $O(n^{2/3}\log n)$, similar to the EMG algorithm. The rest of the cost analysis is as before, and hence we have an $(n^{2/3}\log n,n^{2/3}\log n)$-scheme. 

Thus, we have the following theorem.
\thmsktchprot*
\begin{remark}
We compare the schemes for \Cref{thm:mainthm} and \Cref{thm:sktchprot} (call them Scheme~1 and Scheme 2 respectively). Scheme 2 works for streams of length $m\gg n$ as long as $\Vert \bff\Vert_1=O(n)$, while Scheme 1 requires $m=O(n)$. On the negative side, Scheme 2 has imperfect completeness, contrary to Scheme 1. Furthermore, the space dependence on the error $\eps$ for Scheme 2 is worse than Scheme 1: given any $\eps$, Scheme 2 uses $O(n^{2/3}\log n \log(\eps^{-1}))$ space to bound the completeness and soundness errors by at most $\eps$, while Scheme 1 takes $O\left(n^{2/3}(\log n +\log(\eps^{-1}))\right)$ space to bound the soundness error by $\eps$. This means that to bound the error by $1/\text{poly}(n)$, Scheme 2 takes $O(n^{2/3}\log^2 n)$ space, making it weaker (though simpler) than the scheme of Chakrabarti et al.~\cite{ChakrabartiCMT14}, which takes $O(n^{2/3}\log^{4/3} n)$ space for the same and is also perfectly complete. For this, Scheme 1 takes only $O(n^{2/3}\log n)$ space.    
\end{remark}
\subsection{Special Instances and Applications}\label{sec:apps}
Here, we note important implications of \Cref{thm:mainthm,thm:sktchprot}. They can be applied to get similar results for multiple well-studied problems such as computing the number of distinct elements in the stream ($F_0$), the highest frequency of an element in the stream ($F_{\infty})$, and checking multiset inclusions. Note that for these problems, to the best of our knowledge, the best-known schemes were $(n^{2/3}\log^{4/3} n,n^{2/3}\log^{4/3} n)$-schemes obtained by direct application of the general scheme. Hence, we improve the bounds and simplify the schemes for these problems as well. 

As a direct corollary of \Cref{thm:mainthm,thm:sktchprot}, we get the same bounds for $F_0$. It is an extensively studied problem in both basic streaming and stream verification. It is the special case of frequency-based functions where the function $g$ is defined as $g(x)=0$ if $x=0$, and $g(x)=1$ otherwise. Therefore, we obtain the following result.
\fzero*
Another well-studied problem related to frequency-based functions is computing $F_{\infty}$. Unlike $F_0$, it is not a direct special case, but a protocol for it follows by easily applying a scheme for frequency-based functions. Chakrabarti et al.~\cite{ChakrabartiCMT14} noted one way in which it can be applied to solve $F_{\infty}$. Here, we note a slightly alternate way which doesn't use a subroutine that their scheme uses and is tailored to our protocols: Prover sends the element $j^*\in [n]$ that she claims has the highest frequency and a value $f_{j^*}'$ that she claims to be equal to $f_{j^*}$. By the above protocols, Verifier can check whether $f_{j*}'=f_{j^*}$. If the check passes, he computes $G(\bff):=\sum_{j=1}^n g(f_j)$ using the scheme above, where $g$ is defined as $g(x)=0$ if $x\leq f_{j*}'$ and $g(x)=1$ otherwise. He accepts Prover's claim if $G(\bff)=0$. Thus, we get the following result.
\finf*
The problem of checking multiset inclusion has two multisets arriving in a stream arbitrarily interleaved between each other, and we need to check if one of them is contained in the other. This abstract problem is used as  a subroutine in several other problems, e.g., some graph problems considered in the annotated settings \cite{ChakrabartiCMT14, ChakrabartiG19, ChakrabartiGT20}. Thus, an improved scheme for multiset inclusion implies improved subroutines for the corresponding problems. It can be solved by easy application of frequency-based functions. The reduction is already noted in Chakrabarti et al.~\cite{ChakrabartiCMT14}, but we repeat it here for the sake of completeness.
\multiset*
\begin{proof}
Think of $X$ and $Y$ as $n$-length characteristic vector representation of the multisets (with an entry denoting the multiplicity of the corresponding element). Then, $X\subseteq Y$ if and only if $X_j\leq Y_j$ for each $j\in [n]$. As the elements arrive, we increment an entry if belongs to $Y$ and decrement it if it belongs to $X$. Thus, the vector $\bff$ is given by $f_j = Y_j-X_j$. Define $g$ as $g(x)=0$ if $x\geq 0$ and $g(x)=1$ otherwise. Therefore, computing $G(\bff):=\sum_{j=1}^n g(f_j)$ and checking if it equals $0$ solves the problem. The multisets having size $O(n)$ ensures that the length of the stream is $O(n)$, and so we can safely apply our scheme.  
\end{proof}

\section{Conclusions and Open Problems}
In this work, we designed two new schemes for the broad class of frequency-based functions. These schemes are much simpler than the previously best scheme known for the problem, and furthermore, they even improve upon the complexity bound for space usage and proof-size from $O(n^{2/3}\log^{4/3} n)$ to $O(n^{2/3}\log n)$. The best known upper bound for prescient schemes, given by Chakrabarti et al. \cite{ChakrabartiCMT14}, is $O(n^{2/3}\log n)$ for both of these complexity parameters. Hence, we also close this small gap between the prescient and online scheme complexities, and show that prescience is not necessary here to bring the polylogarithmic factor down to $\log n$. Additionally, the high-level framework used in our schemes is generic enough to be applicable for multiple other problems. 

An important open problem in this area is to determine the asymptotic complexity of the general problem of computing frequency-based functions in the stream verification settings. Chakrabarti et al.~\cite{ChakrabartiCMT14} showed that any online or prescient $(h,v)$-scheme for the problem requires $hv\geq n$. Note that this lower bound leaves open the possibility of a $(\sqrt{n},\sqrt{n})$-scheme, while the best known scheme achieves $(\tO(n^{2/3}),\tO(n^{2/3}))$ for both online and prescient settings. Can we match the lower bound (up to polylogarithmic factors) and get an $(\tO(\sqrt{n}),\tO(\sqrt{n}))$-scheme for the problem, even if prescient? What about for even special cases like $F_0$ or $F_{\infty}$? Recall that there exist such online schemes for the $k$th frequency moment for some constant $k\in \ZZ^+$ \cite{ChakrabartiCMT14}. Also, it is possible to get such a scheme for $F_0$ if we allow multiple rounds of interaction \cite{GurR13}. Any strict improvement on the lower bound would be extremely interesting and a breakthrough. Currently, we don't know of a function in the turnstile streaming model for which any online $(h,v)$-scheme must have total cost $h+v\geq \omega(\sqrt{N})$ where $N$ is the lower bound on its basic streaming complexity. This is related to the major open question of breaking the ``$\sqrt{N}$ barrier'' for the Merlin-Arthur (MA) communication model.

\section*{Acknowledgements}
The author would like to thank Amit Chakrabarti and Justin Thaler for several helpful discussions. He is also grateful to the anonymous FSTTCS 2020 reviewers for their valuable comments, especially to the reviewer who gave a sketch of a scheme using a randomized estimation algorithm, pointing out that it can handle longer streams and that determinism isn't strictly necessary for the subroutine.
\bibliography{refs}

\end{document}

%% file: macros.tex
\usepackage{amsmath,amsthm,amssymb,xspace,verbatim,multirow,bm,bbm,paralist,enumitem}
\usepackage{array,multirow,multicol,booktabs}

\allowdisplaybreaks[1]

\usepackage{algorithm,algorithmicx}
\usepackage[noend]{algpseudocode}

\iflipics
    \usepackage{stmaryrd}
\else
    \usepackage[margin=1in]{geometry}
    \usepackage{txfonts,times}
    \usepackage{paralist}
    \usepackage[T1]{fontenc}
    \usepackage{graphicx}
    \usepackage[usenames,dvipsnames,]{xcolor}
    \usepackage[colorlinks, linkcolor=BrickRed, citecolor=blue]{hyperref}
    \usepackage{thm-restate,cleveref}

    \newtheorem{theorem}{Theorem}[section]
    \newtheorem{lemma}[theorem]{Lemma}

    \theoremstyle{definition}  
    
    \theoremstyle{remark}  \newtheorem*{remark}{Remark}
\fi

\newtheorem{fact}[theorem]{Fact}

\iflipics
    \makeatletter
    \newcommand{\mypar}[1]{\if@nobreak\else\medskip\fi\noindent{\sffamily\bfseries #1.}~}
    \makeatother
\else
    \newcommand{\mypar}[1]{\medskip\noindent{\bfseries #1.}~}
\fi

\newlist{protocol}{description}{1}
\setlist[protocol]{font=\normalfont\em, itemindent=-2ex, itemsep=1pt}

\newcommand{\FF}{\mathbb{F}}

\newcommand{\ZZ}{\mathbb{Z}}

\newcommand{\cA}{\mathcal{A}}

\newcommand{\cH}{\mathcal{H}}

\newcommand{\tf}{\tilde{f}}
\newcommand{\tg}{\tilde{g}}

\newcommand{\hf}{\hat{f}}

\newcommand{\bff}{\mathbf{f}}

\newcommand{\eps}{\varepsilon}
\newcommand{\tO}{\tilde{O}}

\newcommand{\ceil}[1]{\lceil{#1}\rceil}

\renewcommand{\ge}{\geqslant}
\renewcommand{\le}{\leqslant}
\renewcommand{\geq}{\geqslant}
\renewcommand{\leq}{\leqslant}

\DeclareMathOperator{\out}{out}

\newcommand{\eat}[1]{}